\documentclass[11pt,reqno]{amsart}
\usepackage{tikz}
\usepackage{subfig}
\usepackage{epstopdf}
\usepackage{epsfig}
\usepackage{math}
\graphicspath{{./figures/}}
\usepackage{tikz}
\usetikzlibrary{shapes,arrows}
\tikzstyle{decision} = [diamond, draw, fill=blue!20, 
    text width=4.5em, text badly centered, node distance=3cm, inner sep=0pt]
\tikzstyle{block} = [rectangle, draw, fill=blue!20, 
    text width=5em, text centered, rounded corners, minimum height=4em]
\tikzstyle{line} = [draw, -latex']
\tikzstyle{cloud} = [draw, ellipse,fill=red!20, node distance=3cm,
    minimum height=2em]

\newcommand{\rt}{\mathrm{T}}

\newcommand{\rtKL}{\mathrm{TKL}}
\usetikzlibrary{positioning}
\tikzset{main node/.style={circle,fill=blue!20,draw,minimum size=1cm,inner sep=0pt},  }

\begin{document}
\title[Transport alpha divergences]{Transport alpha divergences}
\author[Li]{Wuchen Li}
\email{wuchen@mailbox.sc.edu}
\address{Department of Mathematics, University of South Carolina.}

\begin{abstract}
We derive a class of divergences measuring the difference between probability density functions on the one-dimensional sample space. This divergence is a one-parameter variation of the {Itakura--Saito} divergence between quantile density functions. We prove that the proposed divergence is a one-parameter variation of the transport Kullback-Leibler divergence and the Hessian distance of negative Boltzmann entropy with respect to the Wasserstein-$2$ metric. From Taylor expansions, we also formulate the $3$-symmetric tensor in Wasserstein-$2$ space, which is given by an iterative Gamma three operator. The alpha-geodesic on Wasserstein space is also derived. From these properties, we name the proposed divergences {\em transport alpha divergences}. We provide several examples of {transport alpha divergences on one dimensional distributions, such as generative models and Cauchy distributions.} 
\end{abstract}
\maketitle

\section{Introduction}
{Divergences} between probability distributions play essential roles in statistics, information theory, and signal processing \cite{IG, IG2, CA, CoverThomas1991_elements}. It is a generalization of Kullback-Leibler (KL) divergence, with dualities and variational properties. One typical example is the alpha divergence, which has vast applications in machine learning inference problems and Bayesian sampling problems. 

Information geometry (IG) studies {Legendre dualities}, and invariance properties of divergences in probability density space. Examples include Kullback--Leibler (KL) divergence, {alpha divergences}, and their generalizations \cite{IG, IG2, CA, Zhang}. In this study, the derivatives of negative Boltzmann--Shannon entropy in $L^2$ space play fundamental roles. Its first-order derivative is the likelihood function, its second derivative satisfies the Fisher-Rao metric, while its third derivative introduces the Amari--Chenstov tensor. {These derivatives} characterize the KL divergence and its one-parameter variation, namely {alpha divergences}, where alpha is a scalar. From these characterizations, IG studies and constructs finite-dimensional probability models, with approximation and convexity properties in inference problems. 

Recently, {optimal
transport leads to a different kind of divergences between probability distributions, of which
the Wasserstein distance is a fundamental example \cite{Villani2009_optimal}}. The Wasserstein distance introduces the dualities based {on the ground cost function, known as the Kantorovich duality}. Optimal transport nowadays has vast applications in {artificial intelligence (AI)}, such as generative adversarial networks \cite{ACB}. 
{In particular, there also exists a Riemannian metric for Wasserstein-$2$ distance in probability density space \cite{Lafferty, otto2001, Villani2009_optimal}}. Under the Wasserstein-$2$ metric, the derivatives of negative Boltzman-Shannon entropy are of importance in simulating physics equations \cite{JKO} and Ricci curvature lower bound in a sample space {\cite{BE, NG, LV, RS}}. The study of first- and second-order derivatives in Wasserstein-$2$ space has also been used in statistics and optimization with applications in machine learning algorithms \cite{LiG}. A natural question arises. {{\em Is there an analogue of the alpha divergence based on
optimal transport?}}

{This paper partially addresses this question. We apply information geometry methods, {such as the dualistic
geometry of a divergence/contrast function}, to construct divergences of probability densities from transport maps}. For simplicity of presentation, we focus on the result in one-dimensional sample space, {where the ground distance is the {squared} Euclidean distance}. {We introduce a one-parameter family, named transport alpha divergence}, which interpolates the {transport KL divergence \cite{LiGB1} and the transport Hessian distance \cite{LiGB2}}. We derive the third order derivative, i.e., a $3$--symmetric tensor, of the negative Boltzmann--Shannon entropy in Wasserstein-$2$ space. Several properties of transport {alpha divergence}s are presented, including dualities, Taylor expansions, generalized Bregman divergences, and the generalized Pythagorean theorem in Wasserstein-$2$ space.  

We briefly present the main result. Given a one-dimensional domain $\Omega$ and two {strictly-positive} continuous probability density functions $p$, $q$, we define the {transport alpha divergence} as
\begin{equation}\label{ta}
\mathrm{D}_{\rt, \alpha}(p\|q)=
\left\{\begin{aligned}
&\frac{1}{\alpha^2}\int_0^1\Big(\big(\frac{Q_p'(u)}{Q_q'(u)}\big)^{\alpha}-{\alpha}\log\frac{Q_p'(u)}{Q_q'(u)}-1\Big){du}, &\qquad\textrm{if $\alpha\neq 0$;}\\
&\frac{1}{2}\int_0^1 \left|\log\frac{Q_p'(u)}{Q_q'(u)}\right|^2{du},&\qquad\textrm{if $\alpha=0$.}
\end{aligned}\right.
\end{equation}
where $Q_p$, $Q_q$ are quantile functions of densities $p$, $q$, respectively, and $Q'_p$, $Q'_q$ are derivatives of quantile functions, namely quantile density functions. We note that the quantile function is the inverse function of cumulative distribution function. We remark that compared with {alpha divergence}s in $L^2$ space, the transport {alpha divergence} studies the difference between quantile density functions, instead of probability density functions.  

In literature, several joint studies exist among information geometry, optimal transport, and {alpha divergence}s {\cite{NA, KM, RW, Wong}}. For example, {{{\cite{Wong}} uses OT duality to generalize the Bregman divergence}}. \cite{KM} studies the matrix decomposition viewpoint for different information metrics on Gaussian families. 
{Recently, \cite{NA} also studies Wasserstein-2 metric with general Riemannian ground metric, and then study canniocal divergences from Wasserstein-2 distances.} Compared to the above studies, we focus on the Hessian metric of the negative Boltzman-Shannon entropy in Wasserstein-$2$ space. In this paper, we apply Hessian structures \cite{IG,SY} to construct divergence functionals in terms of {derivatives of transport maps}; see related works in \cite{LiG, LiG1, LiGB1,LiGB2}.  

This paper is organized as follows. In section \ref{sec2}, we briefly review the definition of classical {alpha divergence} in {positive} octant and its relation with information geometry methods. In section \ref{sec3}, we first construct transport {alpha divergence} in one-dimensional sample space, {whose Taylor expansion} shows both the Hessian metric and the 3-symmetric tensor in Wasserstein-$2$ space. Properties of transport {alpha divergence}s, including generalized Bregman divergences and Pythagorean theorem in Wasserstein-$2$ space, are discussed. Several analytical formulas in generative models are provided in section \ref{sec4}.

\section{Divergence functions and Information geometry methods}\label{sec2}
In this section, we briefly review {alpha divergences} in positive octant, and present information geometry methods for studying these divergences {\cite{IG, IG1, IG2}}. {E.g., we mainly follow \cite[section 2.5.2]{IG2} for the definition of alpha connection.} 

Denote a 
$d$--dimensional positive octant by $(0,+\infty)^d$. For any vectors $m=(m_i)_{i=1}^d$, $n=(n_i)_{i=1}^d\in (0, +\infty)^d$, the {alpha divergence} is defined by  
{
\begin{equation*}
\mathrm{D}_\alpha[m\|n]=\sum_{i=1}^d f_\alpha\left(\frac{m_i}{n_i}\right)n_i, 
\end{equation*}}
where $f_\alpha\colon \mathbb{R}_+\rightarrow\mathbb{R}_+$ is a function parameterized by a scalar $\alpha$, such that   
\begin{equation*}
f_\alpha (z)=\left\{\begin{aligned}
\frac{4}{1-\alpha^2}\Big(\frac{1-\alpha}{2}+\frac{1+\alpha}{2}z-z^{\frac{1+\alpha}{2}}\Big), & \quad\alpha\neq\pm 1;\\
 z\log z-(z-1), &\quad \alpha=1; \\ 
 -\log z+(z-1), & \quad \alpha=-1. 
\end{aligned}\right.
\end{equation*}
Here $\log$ is the natural logarithm function. 
The {alpha divergence} is a distance-like function, namely divergence or contrast function that satisfies the following properties. 
{\begin{equation*}
\mathrm{D}_\alpha[m\|n]\geq 0; \quad \mathrm{D}_{\alpha}[m\|n]=0, \quad \textrm{iff}\quad m=n. 
\end{equation*}}
We note that, in general when $\alpha\neq 0$, {alpha divergence} is not a distance function, since {$\mathrm{D}_\alpha[m\|n]\neq \mathrm{D}_\alpha[n\|m]$}. The following dual relation holds 
{
\begin{equation*}
\mathrm{D}_{\alpha}[m\|n]=\mathrm{D}_{-\alpha}[n\|m]. 
\end{equation*}}
There are three important examples of {alpha divergence}s, widely used in statistical inference applications. 
{
\begin{itemize}
\item[(i)] $\alpha=0$: Squared Hellinger distance (up to a scaling factor)
\begin{equation*}
\mathrm{D}_0[m\|n]=2\sum_{i=1}^d(\sqrt{m_i}-\sqrt{n_i})^2.
\end{equation*}
\item[(ii)] $\alpha=1$: Kullback-Leibler (KL) divergence
\begin{equation*}
\mathrm{D}_{1}[m\|n]=\sum_{i=1}^dm_i\log\frac{m_i}{n_i}-(m_i-n_i).
\end{equation*}
\item[(iii)] $\alpha=3$: Chi-squared divergence 
\begin{equation*}
\mathrm{D}_{3}[m\|n]=\frac{1}{2}\sum_{i=1}^d\frac{(m_i-n_i)^2}{n_i}.
\end{equation*}
\end{itemize}}

The {alpha divergence} has several important properties from the Hessian structure of an entropy function, including Taylor expansions and {alpha geodesic}s. Denote a finite dimensional Boltzman-Shannon entropy function by $H(m)=-\sum_{i=1}^n m_i\log m_i$.  Denote the Hessian matrix of negative $H$, also named Fisher matrix, by 
\begin{equation*}
g_{ij}(m)=-\frac{\partial^2}{\partial m_i\partial m_j}H(m)=\frac{1}{m_i}\delta_{ij}, \quad\textrm{for $i,j\in\{1,\cdots, d\}$;}
\end{equation*}
and denote the third derivative of $H$ by a 3--symmetric tensor, known as Amari-Chentsov tensor, 
\begin{equation*}
T_{ijk}(m)=\frac{\partial^3}{\partial m_i\partial m_j\partial m_k}H(m)=\frac{1}{m_i^2}\delta_{ij}\delta_{ik},\quad\textrm{for $i,j,k\in\{1,\cdots, d\}$,}
\end{equation*}
where $\delta_{ij}$ is a Kronecker delta function. The above Hessian matrix and $3$-tensor are useful in {studying the} {alpha divergence}s.
 
Firstly, the Taylor expansion of {alpha divergences} hold:
\begin{equation*}
\begin{split}
\mathrm{D}_\alpha[m\|n]=&\qquad\frac{1}{2}\sum_{i,j=1}^dg_{ij}(n)(m_i-n_i)(m_j-n_j)\\
&+\frac{\alpha-3}{12}\sum_{i,j,k=1}^dT_{ijk}(n)(m_i-n_i)(m_j-n_j)(m_k-n_k)+O(|m-n|^4),
\end{split}
\end{equation*}
where $|\cdot|$ is the Euclidean norm in $\mathbb{R}^d$. Secondly, there are a pair of dual geodesics, namely $\pm$alpha geodesics. {We suggest readers \cite{IG2} for detailed derivations of alpha geodesics and connections.} 
Denote the $\alpha$-connection at a point $m\in \mathbb{R}_+^d$ by a three index {symbol}
\begin{equation*}
 \Gamma^{k, \alpha}_{ij}(m)=-\frac{1+\alpha}{2} m_i\cdot T_{ijk}(m).
\end{equation*}
Then the {alpha geodesic} is given below. Denote $\gamma_\alpha(t)\in\mathbb{R}_+^d$, $t\in[0,1]$, with both initial and terminal points $\gamma_\alpha(0)=m$, $\gamma_\alpha(1)=n$, and 
\begin{equation}\label{aODE}
\frac{d^2}{dt^2}\gamma_\alpha(t)_k+\sum_{i,j=1}^d \Gamma^{k, \alpha}_{ij}(\gamma_\alpha(t)) \frac{d}{dt}\gamma_\alpha(t)_i\frac{d}{dt}\gamma_\alpha(t)_j=0. 
\end{equation}
Note that the above ODE has a closed-form solution after a change variable, namely $\alpha$-representation
\begin{equation}\label{alpha}
k_\alpha(z)=\left\{\begin{aligned}
&\frac{2}{1-\alpha}(z^{\frac{1-\alpha}{2}}-1), &\quad \alpha\neq 1;\\
&\log z, & \quad\alpha=1. 
\end{aligned}\right.
\end{equation}
Hence $\frac{d^2}{dt^2}k_\alpha(\gamma_\alpha(t))=0$. Thus, if $\alpha\neq 1$, the solution of {alpha geodesic} satisfies 
\begin{equation*}
\gamma_\alpha(t)=\Big((1-t) m^{\frac{1-\alpha}{2}}+t n^{\frac{1-\alpha}{2}}\Big)^{\frac{2}{1-\alpha}}.
\end{equation*}
{In above formula, the power and product are
componentwise. In other words, $\gamma_\alpha(t)_k=\Big((1-t) m_k^{\frac{1-\alpha}{2}}+t n_k^{\frac{1-\alpha}{2}}\Big)^{\frac{2}{1-\alpha}}$.}
If $\alpha=-1$, then ODE \eqref{aODE}'s solution is named the mixture (m)-geodesics:
\begin{equation*}
\gamma_{-1}(t)=(1-t)m+t n.
\end{equation*}
If $\alpha=1$, then ODE \eqref{aODE}'s solution is called the exponential (e)-geodesics:
\begin{equation*}
\gamma_1(t)=m^{(1-t)}n^t. 
\end{equation*}
If $\alpha=0$, then \eqref{aODE}'s solution is the Riemannian geodesic of Fisher metric in positive octant:
\begin{equation*}
\gamma_0(t)=\Big((1-t) m^{\frac{1}{2}}+t n^{\frac{1}{2}}\Big)^2.
\end{equation*}
%\end{itemize}

With above defined {alpha geodesic}s, there are duality properties of {alpha divergence}s, including Bregman divergences in terms of $\alpha$-representations \eqref{alpha}, and generalized Pythagorean theorem. In literature \cite{IG1,IG2,SY, Zhang}, $(\mathbb{R}_+^d, g, T)$ is an example of $\alpha$-geometry, or Hessian structure of entropy function $H$. 

{ Two remarks are presented here. The $\alpha$-geometry in positive probability measures are in general different from the alpha geometry in positive measures, where there is a projection to be studied on positive probability measures.}
{In addition, the information geometry method is not limit to the entropy function. See details in \cite{IG2}.}  

\section{Transport {alpha divergence}s}\label{sec5}
In this section, we define {alpha divergence}s in {one dimensional Wasserstein-$2$ space}. Several properties are presented, including composite Bregman divergences and generalized Pythagorean theorem in Wasserstein-$2$ space. We also define the {alpha geodesic} for the completeness of the result.
 
\subsection{Review of Wasserstein-$2$ distances}
We briefly review some basic facts on Wasserstein-$2$ distances. {We only consider the one dimensional case. In this case, the Wasserstein distance will be represented by quantile functions \cite{GD}.} 

Denote a one-dimensional sample space $\Omega=\mathbb{R}^1$ with the {Euclidean distance}. Write a strictly positive probability density space by 
\begin{equation*}
\mathcal{P}(\Omega)=\Big\{p\in C(\Omega)\colon \int_\Omega p(x)dx=1,~p(x)> 0\Big\}.
\end{equation*}
where $\int$, $dx$ are standard integration symbols in $1D$. For any two continuous probability densities $p$, $q\in\mathcal{P}(\Omega)$ with finite second moments, the Wasserstein-$2$ distance \cite{Villani2009_optimal} is defined in the Monge problem by: 
\begin{equation}\label{OT}
\textrm{W}_2(p,q):=\inf_{T\colon \Omega\rightarrow\Omega}\sqrt{\int_\Omega |T(x)-x|^2q(x)dx},
\end{equation}
where the infimum is taken over the continuous transport map function $T$ that pushforwards $q$ to $p$. {When {transport map function $T$} is differentiable}, $T_\#q=p$ means that the Monge-Amper{\'e} equation holds: 
\begin{equation}\label{MA}
p(T(x))\cdot T'(x)=q(x).
\end{equation}

In one-dimensional space, the {optimal transport map function} $T$ is {increasing}, which can be solved analytically in terms of quantile functions. From now on, we denote the cumulative distribution functions (CDFs) $F_p$, $F_q$ of probability density function $p$, $q$, such that 
\begin{equation*}
F_p(x)=\int_{-\infty}^xp(y)dy, \quad F_q(x)=\int_{-\infty}^xq(y)dy. 
\end{equation*}
Denote the quantile functions of probability density $p$, $q$ by 
\begin{equation*}
\begin{split}
Q_p(u)=&\inf\{x\in \mathbb{R}\colon u\leq F_p(x)\}=F_p^{-1}(u),\\
Q_q(u)=&\inf\{x\in \mathbb{R}\colon u\leq F_q(x)\}=F_q^{-1}(u).
\end{split}
\end{equation*}
{Note that $F_p$ and $F_q$ are strictly increasing functions}. We write $F_p^{-1}$, $F_q^{-1}$ are inverse CDFs of $p$, $q$, respectively. We are ready to solve equation \eqref{MA}. Taking the integration on both sides of equation \eqref{MA} w.r.t. $x$, we have
\begin{equation*}
F_p(T(x))=F_q(x).
\end{equation*} 
{Thus, the optimal transport map function} satisfies 
\begin{equation}\label{sMA}
T(x):=F_p^{-1}(F_q(x))=Q_p(F_q(x)). 
\end{equation}
From now on, we always use $T(x)$ to represent {the optimal transport map}. Equivalently, the squared Wasserstein-$2$ distance can be formulated as follows. 
\begin{equation*}
\begin{split}
\textrm{W}_2(p,q)^2=&{\int_\Omega |Q_p(F_q(x))-x|^2q(x)dx}={\int_0^1|Q_p(u)-Q_q(u)|^2du},
\end{split}
\end{equation*}
 where we apply the change of variable $u=F_q(x)\in [0,1]$ in the second equality. In other words, the Wasserstein-$2$ distance in one dimension is the $L^2$ distance in quantile functions.
{In other words, the one-dimensional Wasserstein-$2$ geometry is flat since {$(\mathcal{P}(\Omega), W_2)$} is isometric to a convex subset of $L^2([0,1])$.}

\subsection{Transport {alpha divergence}s}
We are ready to define transport {alpha divergence}. 
Denote a one-parameter function $f_{\mathrm{T},\alpha}\colon \mathbb{R}_+\rightarrow\mathbb{R}_+$ by 
\begin{equation*}
f_{\mathrm{T},\alpha}(z)=\left\{\begin{aligned}
&\frac{1}{\alpha^2}\left(z^{\alpha}-\alpha\log z-1\right),&\quad \mathrm{if}\quad \alpha\neq 0;\\
&\frac{1}{2}|\log z|^2, &\quad \mathrm{if}\quad \alpha=0.
\end{aligned}\right.
\end{equation*}

\begin{definition}[Transport {alpha divergence}]
Define the functional $\mathrm{D}_{\rt,\alpha}\colon \mathcal{P}(\Omega)\times\mathcal{P}(\Omega)\rightarrow\mathbb{R}$ by 
\begin{equation*}
\mathrm{D}_{\rt,\alpha}(p\|q):=\int_\Omega f_{\mathrm{T},\alpha}(T'(x))q(x)dx=\int_\Omega f_{\mathrm{T},\alpha}\left(\frac{q(x)}{p(T(x))}\right)q(x)dx, 
\end{equation*}
where $T$ is the monotone transport map function that pushforwards $q$ to $p$, such that $T_\#q=p$. We name $\mathrm{D}_{\rt, \alpha}$ the transport {alpha divergence}. 
\end{definition}
{From now on, we assume that $f_{\mathrm{T}, \alpha}(T'(x))$ is integrable with the weight function $q(x)$. This ensures that $\mathrm{D}_{\rt,\alpha}(p\|q)<+\infty$.}
We can represent the transport {alpha divergence} in terms of quantile density functions (QDFs). 
Denote the QDFs of probability densities $p$ and $q$ below.  
\begin{equation*}
Q_p'(u)=\frac{d}{du}Q_p(u), \quad Q'_q(u)=\frac{d}{du}Q_{q}(u). 
\end{equation*}
\begin{proposition}
The following equation holds:
\begin{equation}\label{ta1}
\mathrm{D}_{\rt, \alpha}(p\|q)=\int_0^1 f_{\mathrm{T},\alpha}\left(\frac{Q'_p(u)}{Q'_q(u)}\right)du.
\end{equation}
\end{proposition}
\begin{proof}
Denote a variable $u=F_q(x)$, $u\in [0,1]$. Thus, by changing $x$ to $u$ in the following integration, we have
\begin{equation*}
\begin{split}
\int_\Omega f_{\mathrm{T},\alpha}(T'(x))q(x)dx=&\int_\Omega  f_{\mathrm{T},\alpha}\left(\frac{d}{dx}Q_p(F_q(x))\right)q(x)dx\\
=&\int_\Omega  f_{\mathrm{T},\alpha}\Big(\frac{\frac{d}{du}Q_p(u)|_{u=F_q(x)}}{1/\frac{dF_q(x)}{dx}}\Big)q(x)dx\\
=&\int_0^1  f_{\mathrm{T},\alpha}\left(\frac{Q'_p(u)}{Q'_q(u)}\right)du,
\end{split}
\end{equation*}
where the last equality applies the chain rule that $1/\frac{dF_q(x)}{dx}=\frac{dx}{dF_q(x)}=\frac{d}{du}Q_q(u)$. This finishes the proof. 
\end{proof}
\begin{remark}
{ We provide a sufficient condition that $\mathrm{D}_{\rt, \alpha}(p\|q)$ is finite. Assume that there exists constants $C_{p,q}>c_{p,q}>0$, such that $\frac{Q'_p(u)}{Q'_q(u)}\in [c_{p,q}, C_{p,q}]$, for almost everywhere $u\in [0, 1]$. Then $\mathrm{D}_{\rt, \alpha}(p\|q)<+\infty$. One example of this condition is the Cauchy distribution; see Example \ref{ex3}.}
\end{remark}

We next present several examples of transport {alpha divergence}s. 
\begin{itemize}
\item[(i)] $\alpha=1$: transport KL divergence \cite{LiGB1}
\begin{equation*}
\mathrm{D}_{\rt,1}(p\|q)=\int_0^1\Big(\frac{Q_p'(u)}{Q_q'(u)}-\log\frac{Q_p'(u)}{Q_q'(u)}-1\Big)du.
\end{equation*}
\item[(ii)] $\alpha=-1$: transport reverse KL divergence
\begin{equation*}
\mathrm{D}_{\rt,-1}(p\|q)=\int_0^1\Big(\frac{Q_q'(u)}{Q_p'(u)}-\log\frac{Q_q'(u)}{Q_p'(u)}-1\Big)du.
\end{equation*}
\item[(iii)] $\alpha=0$:  transport Hessian distance \cite{LiGB2} (up to a scaling factor)
 \begin{equation*}
\mathrm{D}_{\rt,0}(p\|q)=\frac{1}{2}\int_0^1 \left|\log\frac{Q_p'(u)}{Q_q'(u)}\right|^2du.
\end{equation*}
\end{itemize}
We also present transport {alpha divergence}s with $\alpha=\pm 3$. 
\begin{itemize}
\item[(iv)] $\alpha=3$: transport Chi-square divergence  
\begin{equation*}
\mathrm{D}_{\rt, 3}(p\|q)=\frac{1}{9}\int_0^1\Big(\big(\frac{Q_p'(u)}{Q_q'(u)}\big)^{3}-{3}\log\frac{Q_p'(u)}{Q_q'(u)}-1\Big)du.
\end{equation*}
\item[(v)] $\alpha=-3$: transport inverse Chi-square divergence
\begin{equation*}
\mathrm{D}_{\rt, -3}(p\|q)=\frac{1}{9}\int_0^1\Big(\big(\frac{Q_q'(u)}{Q_p'(u)}\big)^{3}-{3}\log\frac{Q_q'(u)}{Q_p'(u)}-1\Big)du.
\end{equation*}
\end{itemize}

\subsection{Properties}
In this section, we show that there are several dualities and convexity properties for transport {alpha divergence}s. These proofs are based on the facts that transport {alpha divergence}s are generalized Bregman divergences in Wasserstein-$2$ space.  
\begin{proposition}[{Positivity} and Duality]
For any $\alpha\in\mathbb{R}$, and $p$, $q\in\mathcal{P}(\Omega)$, the following properties hold: 
\begin{itemize}
\item[(i)] Positivity: 
\begin{equation*}
\mathrm{D}_{\mathrm{T}, \alpha}(p\|q)\geq 0.  
\end{equation*}
In addition, $\mathrm{D}_{\mathrm{T},\alpha}(p\|q)=0$ if and only if there exists a constant $c\in\mathbb{R}$, such that
\begin{equation*}
 p(x+c)=q(x). 
\end{equation*}
\item[(ii)] Duality: 
\begin{equation*}
\mathrm{D}_{\mathrm{T}, \alpha}(p\|q)=\mathrm{D}_{\mathrm{T}, -\alpha}(q\|p).
\end{equation*}
\end{itemize}
\end{proposition}
\begin{proof}
(i) For $\alpha\neq 0$, note that $x-\log x-1\geq 0$ when $x>0$. Thus,
\begin{equation*}
f_{\mathrm{T}, \alpha}(z)=\frac{1}{\alpha^2}(z^\alpha-\log z^\alpha-1)\geq 0.
\end{equation*}
Since $q>0$, we have $\mathrm{D}_{\mathrm{T}, \alpha}(p\|q)\geq 0$. If $\mathrm{D}_{\mathrm{T}, \alpha}(p\|q)=0$, we have $f_{\mathrm{T},\alpha}(T_p'(x))=0$. Note that $x-\log x-1=0$ iff $x=1$. Thus, $T_p'(x)=1$. This means that $T(x)=x+c$, where $c$ is a constant. From $(T_p)_\#q=p$, we prove (i) with $\alpha\neq 0$. Similarly, we can prove the result for $\alpha\neq 0$. 
 
(ii) The duality is from equation \eqref{ta1}. For any $z_1$, $z_2>0$, we have $f_{\mathrm{T},\alpha}(\frac{z_1}{z_2})=f_{\mathrm{T},-\alpha}(\frac{z_2}{z_1})$. {That is} 
{
\begin{equation*}
\begin{split}
f_{\mathrm{T},\alpha}\left(\frac{z_1}{z_2}\right)=&\left\{\begin{aligned}
&\frac{1}{\alpha^2}\left((\frac{z_1}{z_2})^{\alpha}-\alpha\log \frac{z_1}{z_2}-1\right),&\quad \mathrm{if}\quad \alpha\neq 0;\\
&\frac{1}{2}|\log \frac{z_1}{z_2}|^2, &\quad \mathrm{if}\quad \alpha=0,
\end{aligned}\right.\\
=&\left\{\begin{aligned}
&\frac{1}{\alpha^2}\left((\frac{z_2}{z_1})^{-\alpha}-(-\alpha)\log \frac{z_2}{z_1}-1\right),&\quad \mathrm{if}\quad \alpha\neq 0;\\
&\frac{1}{2}\left|\log \frac{z_2}{z_1}\right|^2, &\quad \mathrm{if}\quad \alpha=0,
\end{aligned}\right.\\
=&f_{\mathrm{T},-\alpha}(\frac{z_2}{z_1}). 
\end{split}
\end{equation*}
}
Replacing $z_1$, $z_2$ by QDFs $Q'_p$, $Q'_q$, respectively, and using \eqref{ta1}, we finish the proof. 

\end{proof}

\begin{proposition}[Taylor expansions in Wasserstein-$2$ space]\label{prop3}
The following equation holds:
\begin{equation*}
\begin{split}
\mathrm{D}_{\rt,\alpha}(p\|q)=&\frac{1}{2}\int_0^1 \left|\frac{Q'_p(u)-Q'_q(u)}{Q'_q(u)}\right|^2du+\frac{\alpha-3}{6}\int_0^1 \Big(\frac{Q'_p(u)-Q'_q(u)}{Q'_q(u)}\Big)^3du\\
&+{O\left(\int_0^1 |\frac{Q'_p(u)-Q'_q(u)}{Q'_q(u)}|^4 du\right)}. 
\end{split}
\end{equation*}
\end{proposition}
\begin{proof}
We note that 
\begin{equation*}
f_{\mathrm{T},\alpha}\left(\frac{Q'_p(u)}{Q'_q(u)}\right)=f_{\mathrm{T},\alpha}(1+h(u)),
\end{equation*}
where we denote a function $h(u):=\frac{Q'_p(u)-Q'_q(u)}{Q'_q(u)}$. By applying a Taylor expansion on $f_{\mathrm{T},\alpha}$ at $1$, we obtain 
\begin{equation*}
f_{\mathrm{T},\alpha}(1+h(u))=f_{\mathrm{T},\alpha}(1)+f'_{\mathrm{T},\alpha}(1)h(u)+\frac{1}{2}f_{\mathrm{T},\alpha}''(1) |h(u)|^2+\frac{1}{6}f_{\mathrm{T},\alpha}'''(1)h(u)^3+O(|h(u)|^4). 
\end{equation*} 
Note that $f_{\mathrm{T},\alpha}(1)=f_{\mathrm{T},\alpha}'(1)=0$, $f_{\mathrm{T},\alpha}''(1)=1$, and $f_{\mathrm{T},\alpha}'''(1)=\alpha-3$. We finish the proof.
\end{proof}

We next represent transport {alpha divergence}s in terms of generalized Bregman divergences in Wasserstein-$2$ space. Denote a function $\mathrm{D}_{-1}=\mathrm{D}_{\mathrm{IS}}\colon \mathbb{R}_+^2\rightarrow\mathbb{R}_+$, such that for $z_1$, $z_2\in\mathbb{R}_+$, 
{\begin{equation*}
\mathrm{D}_{-1}(z_1|z_2):=\frac{z_1}{z_2}-\log\frac{z_1}{z_2}-1. 
\end{equation*}}
The notation {$\mathrm{D}_{-1}$ is a $(-1)$-divergence}, 
which is a Bregman divergence with a potential function 
\begin{equation*}
\Psi(z):=-\log z, \quad z\in\mathbb{R}_+.
\end{equation*}
{Here $\mathrm{D}_{-1}$ is the Itakura--Saito divergence on $\mathbb{R}_+$.}
\begin{theorem}[{Alpha}--Itakura--Saito divergences in Wasserstein-$2$ space]\label{thm1}
Let $\alpha\neq 0$. The following equality holds: 
{\begin{equation*}
\begin{split}
\mathrm{D}_{\rt, \alpha}(p\|q)=&\frac{1}{\alpha^2}\int_0^1 \mathrm{D}_{-1}\left(Q_p'(u)^\alpha | Q_q'(u)^\alpha\right) du. 
\end{split}
\end{equation*}}
In addition, the following generalized Bregman divergence relation holds:
\begin{equation}\label{div}
\begin{split}
\mathrm{D}_{\rt, \alpha}(p\|q)=&\frac{1}{\alpha^2}\int_0^1 \left[\Psi(Q_p'(u)^\alpha)-\Psi\left(Q'_q(u)^\alpha\right)-\Psi'(Q'_q(u)^\alpha)\cdot\left(Q'_p(u)^\alpha-Q'_q(u)^\alpha\right)\right]du. 
\end{split}
\end{equation}
Equivalently, 
\begin{equation}\label{divt}
\begin{split}
\mathrm{D}_{\rt, \alpha}(p\|q)=&\quad\frac{1}{\alpha}\Big[\int_\Omega p(x)\log p(x)dx-\int_\Omega q(x)\log q(x)dx\Big]\\
&+\frac{1}{\alpha^2}\int_\Omega \Big((\frac{q(x)}{p(T(x))})^\alpha-1\Big)q(x)dx. 
\end{split}
\end{equation}
\end{theorem}
\begin{proof}
We first prove equation \eqref{div}. From equation \eqref{ta1}, we have 
\begin{equation*}
\mathrm{D}_{\rt, \alpha}(p\|q)=\int_0^1 f_{\mathrm{T},\alpha}\left(\frac{Q'_p(u)}{Q'_q(u)}\right)du=\frac{1}{\alpha^2}\int_0^1 {\mathrm{D}_{-1}\left(Q_p'(u)^\alpha | Q_q'(u)^\alpha\right)} du. 
\end{equation*}
{From the fact that $\mathrm{D}_{-1}$ is a Bregman divergence function, we have
\begin{equation*}
\mathrm{D}_{-1}(z_1|z_2)=\Psi(z_1)-\Psi(z_2)-\Psi'(z_2)\cdot(z_1-z_2).
\end{equation*}}
This finishes the proof of \eqref{div}. 

We next prove equation \eqref{divt}. Let $u=F_q(x)$, $x=Q_q(u)=F_q^{-1}(u)$. From the chain rule, we have $\frac{dQ_q(u)}{du}=\frac{dx}{dF_q(x)}=\frac{1}{q(x)}$, and  
\begin{equation}\label{proof}
\frac{dQ_p(u)}{du}=\frac{dQ_p(F_q(x))}{dF_q(x)}=\frac{\frac{dQ_p(F_q(x))}{dx}}{\frac{dF_q(x)}{dx}}=\frac{T'(x)}{q(x)}=\frac{1}{p(T(x))},
\end{equation}
where the last equality is from the Monge-Amper{\'e} equation \eqref{MA}. Let us apply the above estimations to equation \eqref{div}. We first observe the following fact. Let $u=F_q(x)$.  
\begin{equation*}
\begin{split}
\int_0^1\Psi\left(Q'_q(u)^\alpha\right)du=&-\alpha\int_0^1\log \left(\frac{d}{du}Q'_q(u)\right)du\\
=&-\alpha\int_0^1\log \frac{1}{q(x)}q(x)dx=\alpha\int_\Omega q(x)\log q(x)dx.
\end{split}
\end{equation*}
Similarly, let $u=F_p(x)$, we have 
\begin{equation*}
\int_0^1\Psi(Q'_p(u)^\alpha)du=\alpha\int_\Omega p(x)\log p(x)dx.
\end{equation*}
We second obtain the following fact. Let $u=F_q(x)$, we have  
\begin{equation*}
\begin{split}
\int_0^1 \Psi'(Q'_q(u)^\alpha)\cdot(Q'_p(u)^\alpha-Q'_q(u)^\alpha) du
=&-\int_0^1 \frac{1}{Q'_q(u)^\alpha}\cdot (Q'_p(u)^\alpha-Q'_q(u)^\alpha) du\\
=&-\int_\Omega \Big((\frac{q(x)}{p(T(x))})^\alpha-1\Big)q(x)dx. 
\end{split}
\end{equation*}
\end{proof}
Following Theorem \ref{thm1}, we note that the transport {alpha divergence} is a Bregman divergence in QDFs after a change of variable. We now present the generalized Pythagorean theorem. 
 Denote the Legendre transformation of function $\Psi(z)=-\log z$ below:
\begin{equation*}
\Psi^*(z^*)=\sup_{z\in\mathbb{R}} \left\{ zz^*-\Psi(z)\right\}.
\end{equation*}
Here $z^*=\Psi'(z)$, and $\Psi^*(z^*)+\Psi(z)=zz^*$. Thus, {$z^*=-\frac{1}{z}$, and $\Psi^*(z^*)=-\log (-z^*)-1$}. 

\begin{corollary}[Generalized Pythagorean theorem in Wasserstein-$2$ space]\label{col:toc}
Let $p$, $q$, $r$ be three probability density functions in $\mathcal{P}(\Omega)$. Assume that the following orthogonal condition holds:
\begin{equation}\label{orth}
\left\{\begin{aligned}
&\frac{1}{\alpha^2}\int_0^1 \big(Q'_p(u)^\alpha-Q'_q(u)^\alpha\big)\cdot\left(\frac{1}{Q'_r(u)^\alpha}-\frac{1}{Q'_q(u)^\alpha}\right) du=0, &\quad \mathrm{if}\quad\alpha\neq 0;\\
&\int_0^1 \log\frac{Q'_p(u)}{Q'_q(u)}\cdot \log\frac{Q'_r(u)}{Q'_q(u)} du=0, &\quad \mathrm{if}\quad\alpha\neq 0. 
\end{aligned}\right.
\end{equation}
Then 
\begin{equation*}
\mathrm{D}_{\mathrm{T}, \alpha}(p\|q)+\mathrm{D}_{\mathrm{T}, \alpha}(q\|r)=\mathrm{D}_{\mathrm{T}, \alpha}(p\|r). 
\end{equation*}
\end{corollary}
\begin{proof}
The proof follows from the definition of Bregman divergences. We note the fact that for $z_1$, $z_2>0$, 
\begin{equation*}
\mathrm{D}_{\mathrm{IS}}(z_1|z_2)=\Psi(z_1)+\Psi^*(z_2^*)-z_1\cdot z_2^*. 
\end{equation*}
Let $\alpha\neq 0$. Denote $K_p=Q'_p(u)^\alpha$ and $K_p^*=-\frac{1}{Q'_p(u)^\alpha}$, for any $p\in\mathcal{P}(\Omega)$. From equation \eqref{div}, we have
\begin{equation*}
\begin{split}
&\mathrm{D}_{\rt, \alpha}(p\|q)+\mathrm{D}_{\rt, \alpha}(q\|r)\\
%=&\frac{1}{\alpha^2}\int_0^1 \Big[\Psi(K_p)-\Psi(K_q)-\Psi'(K_q)\cdot(K_p-K_q)\\
%&\hspace{1cm}+\Psi(K_q)-\Psi(K_r)-\Psi'(K_r)\cdot(K_q-K_r)\Big]du \\
=&\frac{1}{\alpha^2}\int_0^1 \Big[\Psi(K_p)+\Psi^*(K_q^*)-K^*_q\cdot K_p+\Psi(K_q)+\Psi^*(K_r^*)-K^*_r\cdot K_q\Big]du \\
=&\frac{1}{\alpha^2}\int_0^1 \Big[\Psi(K_p)+\Psi^*(K_r^*)-K_p\cdot K^*_r+K_p\cdot K^*_r+K_q\cdot K_q^*-K^*_q\cdot K_p-K^*_r\cdot K_q\Big]du\\
=&\mathrm{D}_{\mathrm{T},\alpha}(p\|r)+\frac{1}{\alpha^2}\int_0^1 (K_p-K_q)\cdot(K^*_r-K^*_q) du. 
\end{split}
\end{equation*}
From the orthogonal condition \eqref{orth}, we finish the proof for $\alpha\neq 0$. For $\alpha=0$, the proof is from the fact that we use the coordinate $\log Q'_p(u)$, under which the transport {alpha divergence} is an Euclidean distance. The result is easy to show. This finishes the proof.   
\end{proof}

We also present the orthogonal condition \eqref{orth} in terms of {transport maps}. 
\begin{corollary}[Transport orthogonal condition]
Orthogonal condition \eqref{orth} is equivalent to
\begin{equation*}
\left\{\begin{aligned}
&\frac{1}{\alpha^2}\int_\Omega \Big(\frac{1}{p(T_p(x))^\alpha}-\frac{1}{q(x)^\alpha}\Big)\cdot \Big(r(T_r(x))^\alpha-q(x)^\alpha\Big)q(x)dx=0, &\quad \mathrm{if}\quad \alpha\neq 0;\\
&\int_\Omega \log\frac{q(x)}{p(T_p(x))}\cdot \log \frac{q(x)}{r(T_r(x))}q(x)dx=0,  &\quad \mathrm{if}\quad \alpha=0,
\end{aligned}\right.
\end{equation*}
where $T_p$, $T_r$ are monotone functions pusforward $q$ to $p$, $r$, respectively. I.e., $(T_p)_\#q=p$, $(T_r)_\#q=r$. 
\end{corollary}
\begin{proof}
We let $u=F_q(x)$. From equation \eqref{proof}, we have 
\begin{equation*}
\frac{dQ_p(u)}{du}=\frac{1}{p(T_p(x))}, \quad \frac{dQ_r(u)}{du}=\frac{1}{r(T_r(x))}.  
\end{equation*}
We finish the proof by substituting the above formulas into condition \eqref{orth}. 
\end{proof}
\begin{remark}
{We remark that the proposed divergences are canonical divergences \cite{AA} associated with the transport Hessian metric. In the appendix, we first review the transport Hessian distance, the transport KL divergence, and then provide a proof of the canonical divergence for the transport alpha divergence.}
\end{remark}
{\begin{remark}
We also point out the other derivation of transport alpha divergences. We work on the potential function $\Psi(z)=-\log z$, $z>0$, and derive the alpha divergence of $\Psi(z)$ on $\mathbb{R}_+$. We then substitute the variable $z$ by the quantile density function in the integral of the domain $[0,1]$. This procedure derives the proposed transport alpha divergence, which is a one family generalization of the transport Hessian distance and the transport KL divergence.   
\end{remark}}

\subsection{Transport {alpha geodesic}}
In this section, we construct a one-parameter family of geodesic equations for quantile density functions. We call them transport {alpha geodesic}s. We also present analytical solutions of transport {alpha geodesic}s. 
 
 \begin{definition}[Transport {alpha geodesic} equations]\label{def2}
Given two probability density functions $p$, $q\in\mathcal{P}(\Omega)$ and $\alpha\in\mathbb{R}$, the transport {alpha geodesic} is defined as below. Denote a transport map function $T_\alpha\colon [0, 1]\times \Omega\rightarrow\Omega$. Consider a one-parameter family of partial differential equations: 
\begin{equation}\label{APDE}
\partial_{tt}\partial_xT_\alpha(t,x)-(\alpha+1)\frac{\big(\partial_{t}\partial_xT_\alpha(t,x)\big)^2}{\partial_x T_\alpha(t,x)}=0,
\end{equation}
with boundary conditions $T_\alpha(0,x)=x$ and $T_\alpha(1,\cdot)_\#q=p$. 
Let the curve $r_\alpha(t,\cdot)\in \mathcal{P}(\Omega)$, $t\in[0,1]$, then 
\begin{equation*}
r_\alpha(t,\cdot)=T_{\alpha}(t,\cdot)_\#q, 
\end{equation*} is the solution of transport {alpha geodesic}.
 \end{definition}
 
\begin{proposition}[Transport {alpha geodesic}s]\label{prop6}
Let $T$ be defined in \eqref{sMA}. Assume $T'(x)\neq 0$ for all $x\in \Omega$. 
A solution of transport {alpha geodesic} is given below. The transport map function $T_\alpha$ satisfies 
\begin{equation*}
\partial_xT_\alpha(t,x)=\left\{\begin{aligned}
&\Big((1-t)+t(T'(x))^{-\alpha}\Big)^{-\frac{1}{\alpha}}, &\quad \textrm{if $\alpha\neq 0$;}\\
&(T'(x))^t, &\quad \textrm{if $\alpha=0$.}
\end{aligned}\right.
\end{equation*}  
Equivalently, denote $r_\alpha(t,\cdot)=T_\alpha(t,\cdot)_\#q$, and write $Q_{r_\alpha}(t,\cdot)$, $\partial_uQ_{r_\alpha}(t,u)$ as the quantile function, quantile density function of probability density {function} $r_\alpha(t,\cdot)$, respectively. Then the transport {alpha geodesic}s in QDFs satisfies
{\begin{equation*}
\partial_uQ_{r_\alpha}(t,u)=\left\{\begin{aligned}
&\Big((1-t)Q'_q(u)^{-\alpha}+t Q'_p(u)^{-\alpha}\Big)^{-\frac{1}{\alpha}},&\quad \textrm{if $\alpha\neq 0$;}\\
&Q'_p(u)^{t}Q'_q(u)^{1-t}, 
&\quad\textrm{if $\alpha=0$.} 
\end{aligned}\right.
\end{equation*}}
\end{proposition}
\begin{proof}
For $\alpha\neq 0$, a simple calculation shows that equation \eqref{APDE} can be reformulated as  
\begin{equation*}
\partial_{tt}\big(\partial_xT_\alpha(t,x)\big)^{-\alpha}=0,
\end{equation*}
with $T_\alpha(0,x)=x$ and $T_\alpha(1,x)=T(x)$. Thus, {the function $\partial_xT_\alpha(t,x)$ is uniquely defined, such that}
%%we have 
\begin{equation*}
\begin{split}
\big(\partial_xT_\alpha(t,x)\big)^{-\alpha}=&t(\partial_xT_\alpha(1,x))^{-\alpha}+(1-t)(\partial_xT_\alpha(0,x))^{-\alpha}\\
=&tT'(x)^{-\alpha}+(1-t).
\end{split}
\end{equation*}
This finishes the first part of the proof. 

{By changing the variable $u=F_q(x)$, we have
\begin{equation*}
\partial_xT_\alpha(t,x)=\partial_xF_{r_\alpha}^{-1}(F_q(x))=Q_{r_\alpha}'(u)\cdot\frac{du}{dx}=\frac{\partial_uQ_{r_\alpha}(t,u)}{Q'_q(u)},
\end{equation*}
for any $t\in [0,1]$. Thus, 
\begin{equation*}
\partial_uQ_{r_\alpha}(t,u)=\Big((1-t)+t\left(\frac{Q'_{p}(u)}{Q'_q(u)}\right)^{-\alpha}\Big)^{-\frac{1}{\alpha}}\cdot Q'_q(u).
\end{equation*}
This finishes the second part of proof. Similar derivations also hold for $\alpha=0$.} 
\end{proof}

Proposition \ref{prop6} can be explained as follows. If $\alpha=-1$, transport $(-1)$ geodesic also satisfies the geodesic equation in Wasserstein-$2$ space, which is ``transportation flat'', meaning that the flatness in the {transport maps}, known as the {McCann's displacement interpolation}:
\begin{equation}\label{tmg}
 \partial_x T_{-1}(t,x)=(1-t)+t\cdot T'(x).
\end{equation}

While, if $\alpha=1$, the transport $1$ geodesic is an ``inverse Jacobi transportation flat'' curve. The transport map function pushforwards the density $q$ to $p$ flatly from the following equation:
\begin{equation}\label{teg}
\partial_xT_{1}(t,x)=\frac{1}{(1-t)+\frac{t}{ T'(x)}}.
\end{equation}
If $\alpha=0$, the transport-$0$ geodesic is a geodesic equation in the transport Hessian metric of negative Boltzmann-Shannon entropy \cite{LiGB2, LiGB3}. 
From now on, we call \eqref{tmg} the {\em$m$--geodesic in Wasserstein-$2$ space}, while name \eqref{teg} the {\em$e$--geodesic in Wasserstein-$2$ space}. 

{
\begin{example}
Consider two Gaussian distributions $p=\mathcal{N}(0, \sigma_p^2)$, $q=\mathcal{N}(0, \sigma_q^2)$, where $\sigma_p$, $\sigma_q>0$ are standard variances of $p$, $q$, respectively. 
Then $T(x)=\sigma_p\sigma_q^{-1}x$. Thus, {$\partial_xT_\alpha(t,x)=\Big((1-t)+t(\sigma_p\sigma_q^{-1})^{-\alpha}\Big)^{-\frac{1}{\alpha}}$. In this case, the transport alpha geodesic $(T_\alpha(t,\cdot))_\#q=r_\alpha(t,\cdot)$ satisfies a time-dependent Gaussian distribution, whose standard variance $\sigma_\alpha$ satisfies $\sigma_\alpha(t)=\sigma_q\cdot \partial_xT_\alpha(t,x)$.}
Hence, let $\alpha=-1, 0, 1$, we have
\begin{equation*}
\partial_xT_{-1}(t,x)=(1-t)+t\sigma_p\sigma_q^{-1}, \quad \partial_xT_0(t,x)=(\sigma_p\sigma_q^{-1})^t,\quad \partial_xT_{1}(t,x)=\frac{1}{(1-t)+t \sigma_p^{-1}\sigma_q}.
\end{equation*}
And 
{
\begin{equation*}
\sigma_{-1}(t)=(1-t)\sigma_q+t\sigma_p, \quad \sigma_0(x)=\sigma_p^t\sigma_q^{1-t},\quad \sigma_1(t)=\frac{1}{(1-t)\sigma_q^{-1}+t \sigma_p^{-1}}.
\end{equation*}}
The following two figures demonstrate the above three transport alpha geodesics. 
 \begin{figure}[H]
    \includegraphics[scale=0.3]{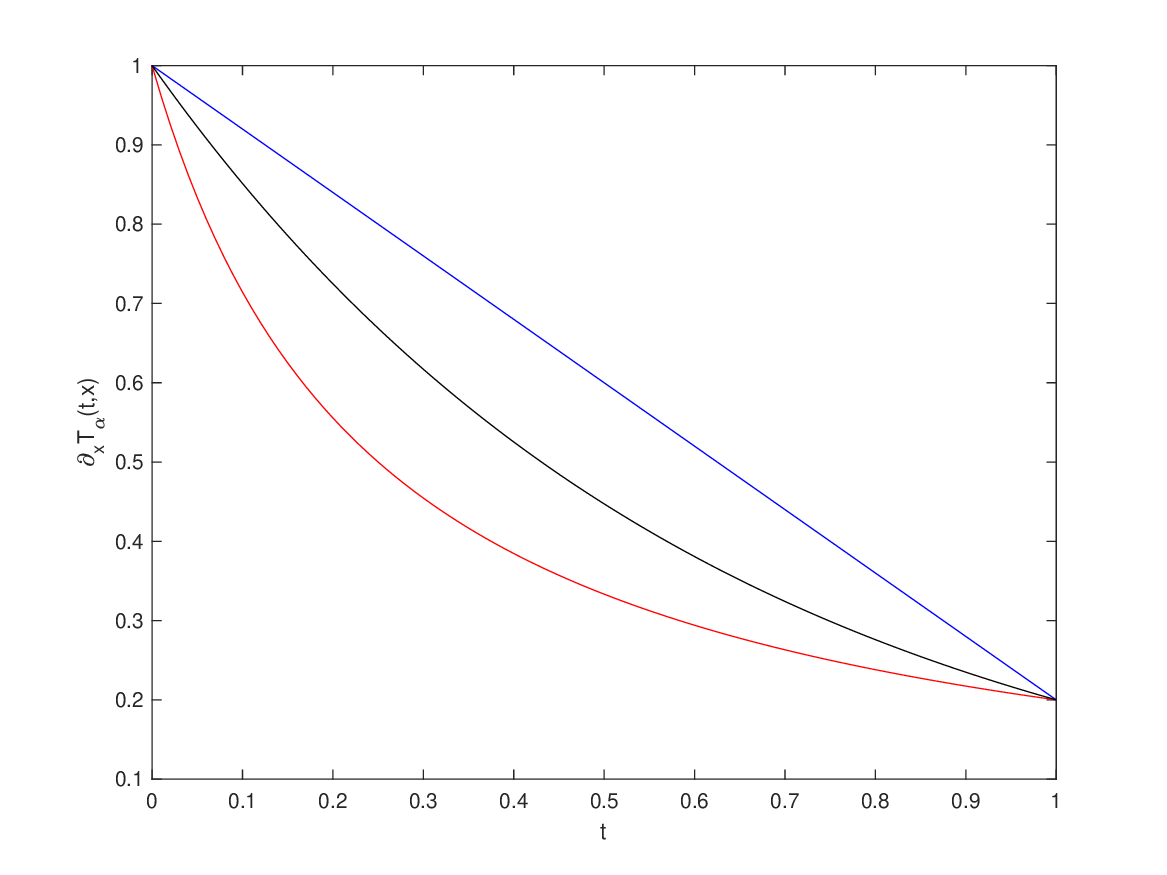} \quad     \includegraphics[scale=0.3]{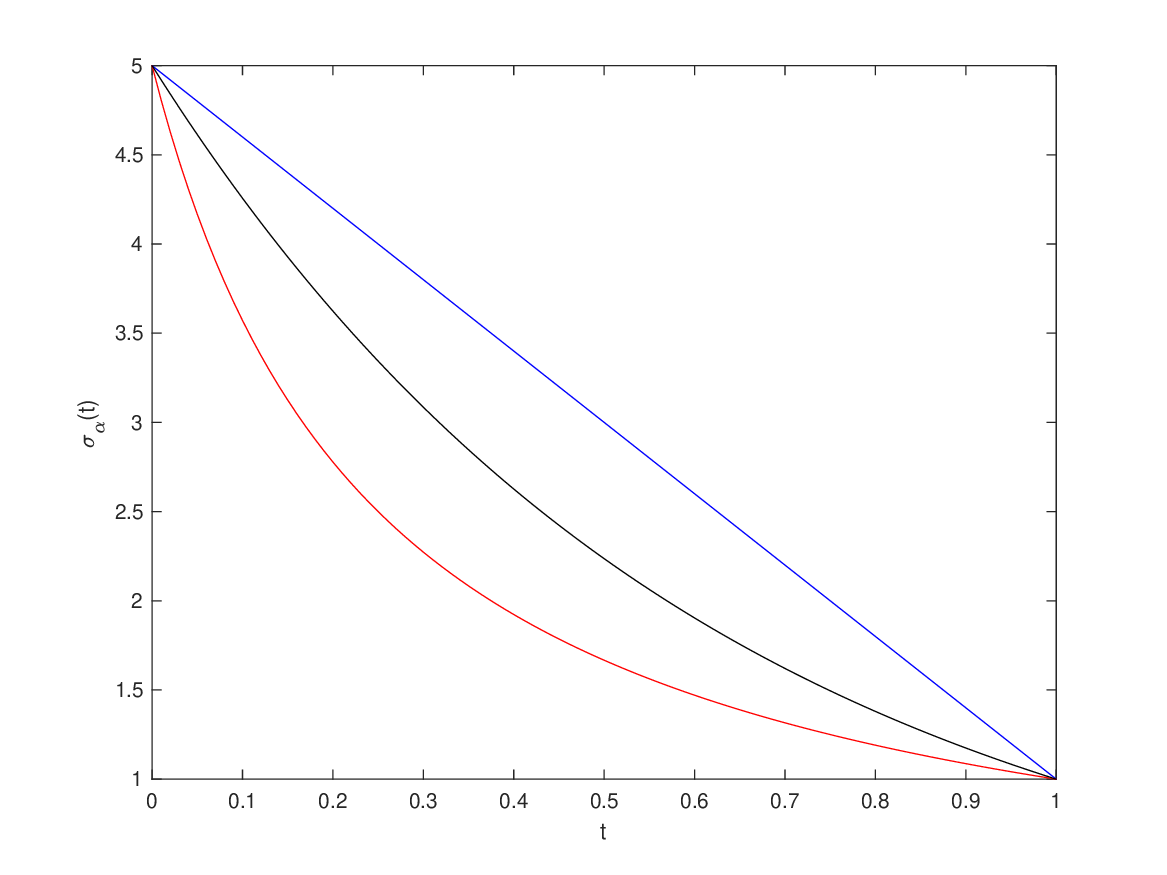}
    \caption{Three curves for $\partial_xT_\alpha(t,x)$ (left), and $\sigma_\alpha(t)$ (right) with $\sigma_p=1$, $\sigma_q=5$. Red: $\alpha=1$. Black: $\alpha=0$. Blue $\alpha=-1$.}
    \label{figure}
\end{figure}
We note that Figure \ref{figure} represents three geodesics in terms of the derivatives of transport map functions and standard variances, which are up to a constant ratio. For standard variances, we show that the transport $(-1)$ geodesic is a line interpolation between $\sigma_p$ and $\sigma_q$, the transport $1$ geodesic is the inverse function of the time variable, while the transport $0$ geodesic is the power function of the time variable.   
\end{example}}
\begin{remark}
{
We informally present the dualistic geometry for transport alpha geodesics. For simplicity of discussion, we study alpha-connections for quantile density functions. Denote a manifold $\mathcal{M}=\Big\{Q'_p\in C^{\infty}[0,1]\colon Q_p'(u)>0,\quad \textrm{for all $u\in [0,1]$}\Big\}$. Denote the tangent vector in the tangent space at $Q'_p\in\mathcal{M}$ as $\mathcal{T}_{Q'_p}\mathcal{M}=\Big\{\dot Q'_p\in C^{\infty}[0,1]\Big\}$. Denote the alpha--connection $\mathbf{\Gamma}^\alpha\colon \mathcal{M}\times C^\infty[0,1]\times C^{\infty}[0,1]\rightarrow C^{\infty}[0,1]$ at the point $Q_p'$ as 
\begin{equation*}
\mathbf{\Gamma}^\alpha:=\mathbf{\Gamma}^\alpha(Q'_p, \dot Q'_p, \dot Q'_p)=-(\alpha+1)\frac{|\dot Q'_p(u)|^2}{Q'_p(u)}.
\end{equation*}
Denote the $0$--connection as $\mathbf{\Gamma}^0$, which is the Levi-Civita connection of the transport Hessian metric:
\begin{equation*}
\mathbf{\Gamma}^0:=\mathbf{\Gamma}^0(Q'_p, \dot Q'_p, \dot Q'_p)=-\frac{|\dot Q'_p(u)|^2}{Q'_p(u)}. 
\end{equation*}
Denote the $1$--connection as $\mathbf{\Gamma}$, which equals to 
\begin{equation*}
\mathbf{\Gamma}:=\mathbf{\Gamma}^{1}(Q'_p, \dot Q'_p, \dot Q'_p)=-2\frac{|\dot Q'_p(u)|^2}{Q'_p(u)}.
\end{equation*}
Denote the $(-1)$--connection as $\mathbf{\Gamma}^*$, such that 
\begin{equation*}
\mathbf{\Gamma}^*:=\mathbf{\Gamma}^{-1}(Q'_p, \dot Q'_p, \dot Q'_p)=0. 
\end{equation*}
The following condition holds: for any $\alpha\in \mathbb{R}$, we have 
$$\mathbf{\Gamma}^{\alpha}=\frac{1+\alpha}{2}\mathbf{\Gamma}+\frac{1-\alpha}{2}\mathbf{\Gamma}^{*}. $$
If $\alpha=0$, then $\mathbf{\Gamma}^0=\frac{1}{2}(\mathbf{\Gamma}+\mathbf{\Gamma}^*)$. If we further denote an operator $\mathbf{C}\colon  \mathcal{M}\times C^\infty[0,1]\times C^{\infty}[0,1]\rightarrow C^{\infty}[0,1]$, such that
$$\mathbf{C}:=\mathbf{C}(Q'_p, \dot Q'_p, \dot Q'_p)=2\cdot \frac{|\dot Q'_p(u)|^2}{Q'_p(u)},$$ then 
\begin{equation*}
\mathbf{\Gamma}^\alpha=\mathbf{\Gamma}^0-\frac{\alpha}{2}\mathbf{C}, \quad \mathbf{\Gamma}^{-\alpha}=\mathbf{\Gamma}^0+\frac{\alpha}{2}\mathbf{C}. 
\end{equation*}
}
\end{remark}
\begin{remark}
{ We remark a relation between the geodesic and the orthogonal condition in Corollary \ref{col:toc}. For $\alpha\neq 0$, note that for $\gamma_\alpha$ connecting $p$, $q$ in Definition \ref{def2}, we have 
\begin{equation*}
\partial_t \Big(\partial_uQ_{r_\alpha}(t,u)\Big)^{-\alpha}=Q'_p(u)^{-\alpha}-Q'_q(u)^{-\alpha}. 
\end{equation*}
Suppose that $l_{-\alpha, p,q}(t,\cdot)$, $t\in[0,1]$ is the transport negative alpha geodesic connecting $p$, $q$, such that $l_{-\alpha,p,q}(0,\cdot)=p(\cdot)$, $l_{-\alpha, p,q}(1,\cdot)=q(\cdot)$.  And $l_{\alpha, r, q}(t,\cdot)$, $t\in [0,1]$, is the transport alpha geodesic connecting $r$, and $q$, such that $l_{\alpha, r, q}(0,\cdot)=r(\cdot)$, $l_{\alpha, r, q}(1)=q(\cdot)$. Then the orthogonal condition in Corollary \ref{col:toc} can be written as 
\begin{equation*}
\int_0^1\partial_t \Big(\partial_uQ_{l_{-\alpha, p,q}}(t,u)\Big)^{\alpha}\cdot \partial_t \Big(\partial_uQ_{l_{\alpha, r,q}}(t,u)\Big)^{-\alpha} du=0.
\end{equation*}
Similar results also hold for $\alpha=0$. 
}
\end{remark}

\section{Hessian structures of entropy in Wasserstein-$2$ space}\label{sec3}
In this section, we formulate the Hessian structures in Wasserstein-$2$ space on one-dimensional sample space.  
In particular, we derive the 3-symmetric tensor from the third order derivatives of negative Boltzmann--Shannon entropy in Wasserstein-$2$ space. \subsection{Review}
We briefly recall some facts about the Wasserstein-$2$ metric \cite{Villani2009_optimal} and the Wasserstein-$2$ Hessian metric \cite{LiGB3}. Denote the smooth, strictly positive probability density space by
\begin{equation*}
\mathcal{P}_o(\Omega)=\Big\{p\in C^{\infty}(\Omega)\colon \int_\Omega p(x)dx=1,~p(x)>0\Big\}.
\end{equation*}
Denote the tangent space at $p\in\mathcal{P}_o(\Omega)$ by
\begin{equation*}
T_p\mathcal{P}_o(\Omega)=\Big\{\sigma\in C^{\infty}(\Omega)\colon \int_\Omega\sigma(x)=0\Big\}.
\end{equation*}
Write the cotangent space at $p\in \mathcal{P}_o(\Omega)$ by 
\begin{equation*}
T_p^*\mathcal{P}_o(\Omega)= C^\infty(\Omega)/\mathbb{R}.
\end{equation*}
For any constant $c\in\mathbb{R}$, if $\Phi\in T_p^*\mathcal{P}_o(\Omega)$, then $\Phi(x)+c\in T_p^*\mathcal{P}_o(\Omega)$. Define an inner product $g_{\mathrm{W}}\colon \mathcal{P}_o(\Omega)\times  T_p\mathcal{P}_o(\Omega)\times  T_p\mathcal{P}_o(\Omega)\rightarrow\mathbb{R}$ by 
\begin{equation*}
g_{\mathrm{W}}(p)(\sigma_1,\sigma_2)=\int_\Omega \Phi'_1(x) \cdot \Phi'_2(x)p(x)dx,
\end{equation*}
where $\sigma_i(x)=-\partial_x\big(p(x)\Phi'_i(x)\big)$, with $\sigma_i\in  T_p\mathcal{P}_o(\Omega)$ and $\Phi_i\in T^*_p\mathcal{P}_o(\Omega)$, for $i=1,2$. Thus, $(\mathcal{P}(\Omega), g_{\mathrm{W}})$ satisfies an infinite-dimensional Riemannian manifold in probability density space. In literature, $(\mathcal{P}(\Omega), g_{\mathrm{W}})$ is often called density manifold \cite{Lafferty} or Wasserstein-$2$ space \cite{otto2001}. 

The Hessian metric in density manifold $(\mathcal{P}(\Omega), g_{\mathrm{W}})$ is defined as follows. Denote the Boltzmann-Shannon entropy by 
\begin{equation*}
\mathcal{H}(p)=-\int_\Omega p(x)\log p(x)dx.
\end{equation*}
Denote the Hessian operator of negative $\mathcal{H}(p)$ by a two form in $(\mathcal{P}(\Omega), g_{\mathrm{W}})$. In other words, let  $g_{\mathrm{H}}=-\textrm{Hess}_{\mathrm{W}}\mathcal{H}\colon \mathcal{P}_o(\Omega)\times  T_p\mathcal{P}_o(\Omega)\times  T_p\mathcal{P}_o(\Omega)\rightarrow\mathbb{R}$, then 
\begin{equation}\label{THM}
g_{\mathrm{H}}(p)(\sigma_1,\sigma_2):=-\textrm{Hess}_{\mathrm{W}}\mathcal{H}(p)(\sigma_1,\sigma_2):=\int_\Omega \Phi''_1(x)\cdot\Phi''_2(x)p(x)dx,
\end{equation}
where $\sigma_i(x)=-\partial_x\big(p(x)\Phi'_i(x)\big)$, with $\sigma_i\in  T_p\mathcal{P}_o(\Omega)$ and $\Phi_i\in T_p^*\mathcal{P}_o(\Omega)$, for $i=1,2$. 
\subsection{Transport 3--symmetric tensor}
We are now ready to formulate the third derivative of negative entropy $-\mathcal{H}(p)$ in Wasserstein-$2$ space. It is a three form, or 3--symmetric tensor in $(\mathcal{P}(\Omega), g_{\mathrm{W}})$.  

\begin{definition}[Transport 3--symmetric tensor]\label{LT}
Denote $T_{\mathrm{H}}\colon \mathcal{P}_o(\Omega)\times  T_p\mathcal{P}_o(\Omega)\times  T_p\mathcal{P}_o(\Omega)\times  T_p\mathcal{P}_o(\Omega)\rightarrow\mathbb{R}$. Then
\begin{equation*}
T_{\mathrm{H}}(p)(\sigma_1, \sigma_2,\sigma_3)=2\int_\Omega \Phi''_1(x)\cdot\Phi''_2(x)\cdot\Phi''_3(x) p(x)dx,
\end{equation*}
where {$\sigma_i(x)=-\big(p(x)\Phi'_i(x)\big)'$}, with $\sigma \in  T_p\mathcal{P}_o(\Omega)$, and $\Phi_i\in T^*_p\mathcal{P}_o(\Omega)$, for $i=1,2, 3$.
\end{definition}
We also present that the transport 3-symmetric tensor introduces a third-order iterative Bakry--{\'E}mery Gamma calculus; {see related studies between Gamma calculuses \cite{BE} and optimal transport \cite{LV, RS}}.  
\begin{theorem}[Gamma calculus induced 3--symmetric tensor]
Denote bilinear forms $\Gamma_1$, $\Gamma_2\colon C^{\infty}(\Omega)\times C^{\infty}(\Omega)\rightarrow C^\infty(\Omega)$ by 
\begin{equation*}
\Gamma_1(\Phi, \Phi)(x)=\Phi'(x)\cdot \Phi'(x),\quad \Gamma_2(\Phi, \Phi)(x)=\Phi''(x)\cdot \Phi''(x).
\end{equation*}
Define the Gamma-3 operator $\Gamma_3\colon C^{\infty}(\Omega)\times C^{\infty}(\Omega)\times C^{\infty}(\Omega)\rightarrow C^\infty(\Omega)$ by 
\begin{equation*}
\Gamma_3(\Phi,\Phi, \Phi)(x):=\Gamma_2(\Gamma_1(\Phi, \Phi), \Phi)(x)-\Gamma_1(\Gamma_2(\Phi,\Phi), \Phi)(x).
\end{equation*}
Then the following equation holds:
 \begin{equation*}
T_{\mathrm{H}}(p)(\sigma, \sigma, \sigma)=\int_\Omega \Gamma_3(\Phi(x), \Phi(x), \Phi(x))p(x)dx,
\end{equation*}
where $\sigma=-\partial_x(p(x)\Phi'(x))$.
\end{theorem}
\begin{proof}The proof follows by a direct calculation. Note that 
\begin{equation*}
\Gamma_1(\Gamma_2(\Phi,\Phi), \Phi)=\partial_x(|\Phi''|^2)\Phi'=2\Phi'''\cdot\Phi''\cdot\Phi',
\end{equation*}
and 
\begin{equation*}
\Gamma_2(\Gamma_1(\Phi,\Phi), \Phi)=\partial_x^2(|\Phi'|^2)\Phi''=2\Phi'''\cdot\Phi''\cdot\Phi'+2|\Phi''|^3.
\end{equation*}
By taking the difference between the two functionals, we derive the result. 
\end{proof}

We finish this section by representing Taylor expansions of transport {alpha divergence}s, using the Hessian structure $(\mathcal{P}_o(\Omega), g_{\mathrm{H}}, T_{\mathrm{H}})$. 
\begin{corollary}[Taylor expansions in transport Hessian structures]
For any $p$, $q\in\mathcal{P}_o(\Omega)$. Denote $\Phi\in T^*_q\mathcal{P}_o(\Omega)$, such that
\begin{equation*}
\Phi(x)=\int_0^x Q_p(F_q(y))dy-\frac{x^2}{2}+c, 
\end{equation*}
where $c\in\mathbb{R}$ is a constant. Denote $\sigma=\partial_x(q(x)\Phi'(x))\in T_q\mathcal{P}_o(\Omega)$. Then, the following equation holds. 
\begin{equation*}
\begin{split}
\mathrm{D}_{\rt,\alpha}(p\|q)=&\frac{1}{2}g_{\mathrm{H}}(q)(\sigma, \sigma)+\frac{\alpha-3}{6}T_{\mathrm{H}}(q)(\sigma, \sigma, \sigma)+{O\left(\int_\Omega |\Phi''(x)|^4 q(x)dx\right)}. 
\end{split}
\end{equation*}
\end{corollary}
\begin{proof}
The proof is based on a direct calculation. Note that 
\begin{equation*}
\Phi'(x)=Q_p(F_q(x))-x,
\end{equation*}
and 
\begin{equation*}
\Phi''(x)=\frac{d}{dx}Q_p(F_q(x))-1=\frac{Q'_p(F_q(x))}{\frac{1}{q(x)}}-1. 
\end{equation*}
For $k=2,3$, from the change of variable $u=F_q(x)$, we have
\begin{equation*}
\int_\Omega (\Phi''(x))^kp(x)dx=\int_0^1 \left(\frac{Q'_p(u)}{Q'_q(u)}-1\right)^k du. 
\end{equation*}
From Proposition \ref{prop3}, we finish the proof. 
\end{proof}
\begin{remark}
We note that $\Gamma_1$, $\Gamma_2$ are often called Gamma one and Gamma two operators, which are firstly introduced by Bakry--{\'E}mery \cite{BE} to study the Ricci curvature lower bound on a sample space. Here we only show them in one-dimensional sample space. The iterative Gamma two calculus connects with second-order derivatives of entropy in Wasserstein-$2$ space \cite{BE, Villani2009_optimal} with generalizations \cite{LiG}. Here, we present a ``third-order'' Gamma calculus to formulate the third derivatives in Wasserstein-$2$ space, namely {\em transport 3-symmetric tensor}. We will study geometric calculations of transport-3 symmetric tensors in high-dimensional spaces in future works. Following \cite{LiG, LiG1}, we expect that the information geometry method and Gamma three operators are tools in studying generalized divergences in high dimensional Wasserstein-$2$ spaces.  
\end{remark}
\section{Examples}\label{sec4}
This section provides examples of transport {alpha divergence}s between one-dimensional probability distributions, including generative models, location-scale families, and Cauchy distributions. 

 In machine learning applications \cite{ACB}, a generative model is defined as follows. Consider a latent random variable $Z\sim p_{\mathrm{ref}}$, where $p_{\mathrm{ref}}\in\mathcal{P}(\Omega)$ is a given reference measure. Denote a map function $G\colon \Omega\times \Theta\rightarrow\Omega$, where $\Theta\subset\mathbb{R}^n$ is a parameter space. { We also assume that $G(\cdot, \theta)$ is a monotone mapping for all $\theta\in \Theta$.} Then 
\begin{equation*}
 G(\cdot, \theta)_\#p_{\textrm{ref}}(\cdot)=p(\cdot,\theta).
\end{equation*}
If $G$ is linear w.r.t. $Z$, the generative family forms a location-scale family. Furthermore, if $G$ is linear and $Z$ follows a Gaussian distribution, the generative model formulates a class of Gaussian distributions. 
\begin{proposition}[Transport {alpha divergence} in one dimensional distributions]
Let $\theta_X$, $\theta_Y\in\Theta$ and consider $Z\sim p_{\mathrm{ref}}$, with  
\begin{equation*}
X=G(Z, \theta_X)\sim p_X,\quad Y=G(Z, \theta_Y)\sim p_Y. 
\end{equation*}
Then the transport {alpha divergence} between probability distributions $p_X$, $p_Y$ satisfies 
\begin{equation*}
\mathrm{D}_{\rt, \alpha}(p_X\|p_Y)=\left\{\begin{aligned}
&\frac{1}{\alpha^2}\mathbb{E}_{Z\sim p_{\mathrm{ref}}}\Big[\Big(\frac{\partial_{Z}G(Z, \theta_X)}{\partial_{Z}G(Z, \theta_Y)}\Big)^\alpha-\alpha\log\frac{\partial_{Z}G(Z, \theta_X)}{\partial_{Z}G(Z, \theta_Y)}-1\Big],&\quad \textrm{if $\alpha\neq 0$;}\\
&\frac{1}{2}\mathbb{E}_{Z\sim p_{\mathrm{ref}}}\Big[\Big(\log\frac{\partial_{Z}G(Z, \theta_X)}{\partial_{Z}G(Z, \theta_Y)}\Big)^2\Big],&\quad\textrm{if $\alpha=0$.}
\end{aligned}\right.
\end{equation*}
Here $\mathbb{E}$ is the expectation operator. We also compare transport {alpha divergence}s with the Wasserstein-$2$ distance 
\begin{equation*}
W_2(p, q)=\sqrt{\mathbb{E}_{Z\sim p_{\mathrm{ref}}}\Big[\big|G(Z, \theta_X)-G(Z, \theta_Y)\big|^2\Big]}, 
\end{equation*}
where we need to assume that $\mathbb{E}_{Z\sim p_{\mathrm{ref}}} |G(Z, \theta)|^2<+\infty$, for $\theta=\theta_X$ or $\theta_Y$. {We note that if $\alpha\neq 0$, the transport alpha divergence is not symmetric on parameters $\theta_X$, $\theta_Y$. This is in contrast with the Wasserstein-2 distance.}
\end{proposition}
{\begin{remark} 
One {may} apply the transport alpha divergence to measure the closeness between two neural network parameters. The other application is to {apply transport alpha divergences as loss functions for inference} problems. Some related studies have been conducted for the quantile density function in statistics \cite{Par}. 
However, in general high-dimensional sample spaces, the construction of transport alpha divergences, with the selection of parameter alpha in inference problems, are unclear, which are left for the future work. 
\end{remark}
\begin{remark}
{The proposed transport alpha divergence is not intended as a replacement formula for the Wasserstein-$2$ distance. Instead, it modifies the geometric structure of the density manifold. While the local expansion of the Wasserstein-$2$ distance yields only the Riemannian metric (the Wasserstein information matrix), the Taylor expansion of the transport alpha divergence contains an additional third-order operator, which beyonds the transport Hessian metric. At present, the optimization properties of the divergence as a loss function are not fully understood. The systematic analysis with the comparison with the Wasserstein-$2$ distance loss function, will be investigated in future work.}
\end{remark}
\begin{example}[Location scale family]
Suppose $G$ is a linear transport map function such that 
\begin{equation*}
G(Z, \theta)=\theta Z,
\end{equation*}
with $\theta>0$ and $Z\in\mathbb{R}^1$. 
Then $p(\cdot,\theta)=G(\cdot, \theta)\# p_{\textrm{ref}}$ is a location scale family. In this case, we have 
\begin{equation*}
\mathrm{D}_{\rt, \alpha}(p_X\|p_Y)=\left\{\begin{aligned}
&\frac{1}{\alpha^2}\Big[\Big(\frac{\theta_X}{\theta_Y}\Big)^\alpha-\alpha\log\frac{\theta_X}{\theta_Y}-1\Big],&\quad \textrm{if $\alpha\neq 0$;}\\
&\frac{1}{2}\Big(\log\frac{\theta_X}{\theta_Y}\Big)^2,&\quad\textrm{if $\alpha=0$.}
\end{aligned}\right.
\end{equation*}
 \end{example}

We last present an example of the Wasserstein-$2$ distance not being well defined, meaning that the distributions are not with the finite second moment. In this case, the transport {alpha divergence} is still well defined.  
\begin{example}[Cauchy distributions]\label{ex3}
The Cauchy distribution is defined as follows. For $\gamma>0$, 
\begin{equation*}
p(x, \gamma)=\frac{1}{\pi \gamma}\Big[\frac{1}{(\frac{x}{\gamma})^2+1}\Big]. 
\end{equation*}
Thus, denote $T(x)=\gamma\cdot x$, we have $T_\#p(\cdot, 1)=p(\cdot, \gamma)$. For $\gamma_1$, $\gamma_2>0$, we have 
\begin{equation*}
\mathrm{D}_{\rt, \alpha}(p(\cdot, \gamma_1)\| p(\cdot, \gamma_2))=\left\{\begin{aligned}
&\frac{1}{\alpha^2}\Big[\Big(\frac{\gamma_1}{\gamma_2}\Big)^\alpha-\alpha\log\frac{\gamma_1}{\gamma_2}-1\Big],&\quad \textrm{if $\alpha\neq 0$;}\\
&\frac{1}{2}\Big(\log\frac{\gamma_1}{\gamma_2}\Big)^2, &\quad \textrm{if $\alpha=0$.}
\end{aligned}\right.
\end{equation*}
{We remark that the Cauchy distribution does not have the finite 
second moment.} Thus, the Wasserstein-$2$ distance is not well defined, i.e., {$\mathrm{W}_2(p(\cdot, \gamma_1), p(\cdot, \gamma_2))=+\infty$}. 
\end{example}

\section{Discussion}
This paper proposes {a class of transport alpha divergence}s, one-parameter variation of {the} transport KL divergence and {the} transport Hessian distance. They are connected with Hessian metrics and $3$--symmetric tensors of the negative Boltzmann-Shannon entropy in Wasserstein-$2$ space. We provide several analytical examples in one-dimensional probability densities, including generative models and {Cauchy distributions}. 

It is worth mentioning that the quantile density functions (QDFs) have been applied in statistical learning problems \cite{Par}. The quantile density functions measure densities' shape up to any constant shifting. The transport {alpha divergence} provides a class of functionals for measuring the difference from QDFs, i.e., {derivatives of transport map functions}. In future work, we shall study transport alpha divergences in high dimensional probability densities \cite{CA,LiGB3}. This direction includes analysis, dualities, invariance properties, and optimization algorithms of transport mapping-related divergence functionals. In particular, systematic geometric calculations for Hessian structures in Wasserstein-$2$ space $(\mathcal{P}_o(\Omega), g_{\mathrm{H}}, T_{\mathrm{H}})$ will be investigated; see related studies in {\cite{CY, LiGB3, LiGB4, SY}}. {The {systematic} understanding of three symmetric tensors and alpha connections rely on both studies in optimal transport and information geometry; see related studies in \cite{NA, AA, RW}. 

\noindent\textbf{Acknowledgements}. {W. Li's work is supported by AFOSR YIP award No. FA9550-23-1-0087, NSF RTG: 2038080, NSF DMS: 2245097 and the McCausland Faculty Fellow in University of South Carolina.}

\newpage
\section*{Appendix}
{In this section, we provide some necessary proofs, which are used in this paper. }
%some calculations of high-order derivatives of entropy in Wasserstein-$2$ space.  
\subsection{Derivatives in Wasserstein-$2$ space}
We first present first, second, and third-order derivatives in Wasserstein-$2$ space. This provides the derivation for transport $3$--symmetric tensor defined in Definition \ref{LT}.
\begin{proposition}
Denote $p\colon [0,1]\times\Omega\rightarrow\mathbb{R}$ satisfying the geodesics equation in $(\mathcal{P}_o(\Omega), g_{\mathrm{W}})$ with $p(0,x)=p(x)$, $\partial_tp(0,x)=\sigma(x)=-\partial_x(p(x)\Phi'(x))$. Then 
\begin{equation*}
-\frac{d^n}{dt^n}\mathcal{H}(p(t,\cdot))=(-1)^n (n-1)!\int_\Omega (\Phi''(x))^n p(x)dx.
\end{equation*}
In particular, for $n=1,2,3$, we have 
\begin{itemize}
\item[(i)] 
\begin{equation*}
\frac{d}{dt}\mathcal{H}(p(t,\cdot))|_{t=0}=\mathrm{grad}_{\mathrm{W}}\mathcal{H}(p)(\sigma)=\int_\Omega \Phi''(x)p(x)dx.
\end{equation*}
\item[(ii)]
\begin{equation*}
\frac{d^2}{dt^2}\mathcal{H}(p(t,\cdot))|_{t=0}=\mathrm{Hess}_{\mathrm{W}}\mathcal{H}(p)(\sigma,\sigma)=\int_\Omega (\Phi''(x))^2p(x)dx.
\end{equation*}
\item[(iii)]
\begin{equation*}
\frac{d^3}{dt^3}\mathcal{H}(p(t,\cdot))|_{t=0}=T_{\mathrm{H}}(p)(\sigma, \sigma, \sigma)=2\int_\Omega (\Phi''(x))^3 p(x)dx.
\end{equation*}
\end{itemize}
\end{proposition}
\begin{proof}
We recall that the {Levi-Civita connection} induced geodesics in $(\mathcal{P}_o(\Omega), g_{\mathrm{W}})$ satisfies  
\begin{equation*}
\left\{\begin{aligned}
&\partial_{t} p(t,x)+\partial_x(p(t,x)\partial_x \Phi(t,x))=0\\
&\partial_{t} \Phi(t,x)+\frac{1}{2}|\partial_x\Phi(t,x)|^2=0,
\end{aligned}\right.
\end{equation*}
where $p(0,x)=p(x)$ and $\partial_tp(0,x)=\sigma(x)=-\partial_x(p(x)\Phi'(x))$. 

We prove the result by induction. When $n=1$, we have
\begin{equation*}
\begin{split}
-\frac{d}{dt}\mathcal{H}(p(t,\cdot))|_{t=0}=&-\int_\Omega \partial_x(p(x)\Phi'(x)) (\log p(x)+1) dx\\
=&\int_\Omega \Phi'(x)\partial_x\log p(x)p(x)dx\\
=&\int_\Omega \Phi'(x) \partial_x p(x)dx\\
=&-\int_\Omega \Phi''(x) p(x) dx,
\end{split}
\end{equation*}
where we use the fact that $\partial_x\log p(x)\cdot p(x)=\frac{\partial_xp(x)}{p(x)}\cdot p(x)=\partial_xp(x)$ in the third equality. Assume that for $n=k$, $k\in \mathbb{N}$, we have 
\begin{equation*}
-\frac{d^k}{dt^k}\mathcal{H}(p(t,\cdot))|_{t=0}=(-1)^k(k-1)!\int_\Omega (\Phi''(x))^k p(x)dx.
\end{equation*}
Note that the second equation of the geodesic in $(\mathcal{P}_o(\Omega), g_{\mathrm{W}})$ can be reformulated as below:
\begin{equation*}
\partial_t\partial_x\Phi(t,x)+\partial_{xx}\Phi(t,x)\cdot\partial_x\Phi(t,x)=0. 
\end{equation*}
Hence
\begin{equation*}
\begin{split}
\frac{d}{dt}\int_\Omega (\partial_{xx}\Phi(t,x))^k p(t,x)dx=&\int_\Omega\partial_t\Big((\partial_{xx}\Phi(t,x))^k\Big) p(t,x)dx+\int_\Omega (\partial_{xx}\Phi(t,x))^k\partial_tp(t,x)dx\\
=&\quad\int_\Omega k(\partial_{xx}\Phi(t,x))^{k-1} \partial_x^2\partial_{t}\Phi(t,x)p(t,x)dx\\
&-\int_\Omega (\partial_{xx}\Phi(t,x))^k \partial_x(p(t,x)\partial_x\Phi(t,x))dx\\
=&-\int_\Omega k(\partial_{xx}\Phi(t,x))^{k-1} \partial_x(\partial_{xx}\Phi(t,x)\partial_x\Phi(t,x))p(t,x)dx\\
&+\int_\Omega k(\partial_{xx}\Phi(t,x))^{k-1}\partial_x^3\Phi(t,x)\partial_x\Phi(t,x)p(t,x)dx\\
=&-\int_\Omega k(\partial_{xx}\Phi(t,x))^{k+1}p(t,x)dx.
\end{split}
\end{equation*}
 From the assumption, we have
\begin{equation*}
\begin{split}
-\frac{d^{k+1}}{dt^{k+1}}\mathcal{H}(p(t,\cdot))|_{t=0}=&(-1)^k(k-1)!\frac{d}{dt}\int_\Omega (\Phi''(t,x))^k p(t,x)dx|_{t=0}\\
=&(-1)^k(k-1)!\cdot (-1)\cdot k\int_\Omega (\Phi''(x))^{k+1}p(x)dx\\
=&(-1)^{k+1}k!\int_\Omega (\Phi''(x))^{k+1}p(x)dx,
\end{split}
\end{equation*}
which finishes the proof. 
\end{proof}
\begin{remark}
These geometric formulas are derived based on the Riemannian Levi-Civita connection in density manifold $(\mathcal{P}_o(\Omega), g_{\mathrm{W}})$. They formulate classical Gamma calculuses; see details in \cite{LiG, LiG1, Villani2009_optimal}. We leave the studies of high-order derivatives of negative entropy in $(\mathcal{P}_o(\Omega), g_{\mathrm{W}})$ in high dimensional sample spaces in future works.  
\end{remark}
\subsection{Transport Hessian distance}
{We next review the definition of transport Hessian distances between probability distributions. For completeness of this paper, we also provide its derivation here. %see generalization of Hessian distances in   

%\begin{definition}[Transport information Hessian distance \cite{LiHess}]
 Define a distance function $\mathrm{Dist}_{\mathrm{H}}$ $\colon \mathcal{P}(\Omega)\times\mathcal{P}(\Omega)\rightarrow\mathbb{R}$ by
\begin{equation}\label{HC}
\mathrm{Dist}_{\mathrm{H}}(p,q)^2=\inf_{p\colon [0,1]\times\Omega\rightarrow\mathbb{R}}\Big\{\int_0^1 g_\mathrm{H}(\partial_t p, \partial_t p)dt\colon p(0,x)=q(x),~p(1,x)=p(x)\Big\}.
\end{equation}
Here the infimum is taken among all smooth density paths $p\colon [0, 1]\times \Omega\rightarrow\mathbb{R}$, which connects both initial and terminal time probability density functions $q$, $p\in \mathcal{P}(\Omega)$. {Interestingly, the variational problem \eqref{HC} admits a closed-form expression.}
%\end{definition}
%\subsection{Formulations}
%%We first derive closed-form solutions for transport information Hessian distances defined by \eqref{HC}. 
\begin{proposition}[\cite{LiGB2}].
%%Denote a one dimensional function $h\colon \Omega\rightarrow\mathbb{R}$ by
%%\begin{equation*}
%%h(y)=\log y.
%%\end{equation*}
The squared transport Hessian distance has the following formulation. 
\begin{equation*}
\mathrm{Dist}_{\mathrm{H}}(p, q)^2=\int_0^1 \left|\log\frac{Q_p'(u)}{Q_q'(u)}\right|^2du=2\mathrm{D}_{\rt,0}(p\|q).
\end{equation*}
\end{proposition}
\begin{proof}
 There are two change of variables to derive the Hessian distance. Firstly, denote $p^0(x)=q(x)$ and $p^1(x)=p(x)$. Denote the variational problem \eqref{HC} by 
\begin{equation*}
\begin{split}
\mathrm{Dist}_{\mathrm{H}}(p^0,p^1)^2=\inf_{\Phi,p\colon[0,1]\times \Omega\rightarrow\mathbb{R}}&\Big\{\int_0^1\int_\Omega |\partial^2_{yy}\Phi(t,y)|^2p(t,y)dydt\colon \\
&\hspace{1cm}\partial_tp(t,y)+\partial_y(p(t,y) \partial_y\Phi(t,y))=0,~\textrm{fixed $p^0$, $p^1$}\Big\},
\end{split}
\end{equation*}
where the infimum is among all smooth density paths $p\colon [0, 1]\times\Omega\rightarrow\mathbb{R}$ satisfying the continuity equation with the gradient potential vector field. The potential is given by $\Phi\colon [0, 1]\times\Omega\rightarrow\mathbb{R}$.  Denote 
\begin{equation*}
y=T(t,x),\qquad \partial_tT(t,x)=\partial_y\Phi(t,y).
\end{equation*}
Hence 
\begin{equation*}
\begin{split}
\mathrm{Dist}_{\mathrm{H}}(p^0,p^1)^2=&\inf_{T\colon[0,1]\times\Omega\rightarrow \Omega}\Big\{\int_0^1\int_\Omega|\partial_y v(t,T(t,x))|^2p(t,T(t,x))dT(t,x)dt\colon\\
&\hspace{3cm} T(t,\cdot)_\#p(0,\cdot)=p(t,\cdot)\Big\},
\end{split}
\end{equation*}
where the infimum is taken among all smooth transport map functions $T\colon [0,1]\times \Omega \rightarrow \Omega$ with $T(0,x)=x$ and $T(1,x)=T(x)$. Thus, 
%We observe that the above variation problem leads to 
\begin{equation}\label{variation1}
\begin{split}
&\int_0^1\int_\Omega|\partial_y \partial_tT(t,x)|^2p(t,T(t,x))\partial_xT(t,x)dxdt\\
=&\int_0^1\int_\Omega|\partial_x \partial_tT(t,x)\frac{dx}{dy}|^2p(t,T(t,x))\partial_xT(t,x)dxdt\\
=&\int_0^1\int_\Omega|\partial_t \partial_xT(t,x)\frac{1}{\partial_xT(t,x)}|^2{q(x)}dxdt\\
=&\int_0^1\int_\Omega|\partial_t \log \partial_xT(t,x)|^2{q(x)}dxdt,
%%=&\int_0^1\int_\Omega|\partial_t\partial_x T(t,x)\frac{1}{(\partial_xT(t,x))^{3/2}}\sqrt{f''(\frac{q(x)}{\partial_xT(t,x)})}|^2q(x)^2dxdt.
\end{split}
\end{equation}
where we use the fact that $\partial_t \log \partial_xT(t,x)=\partial_t \partial_xT(t,x)\frac{1}{\partial_xT(t,x)}$. 

Secondly, denote $y=F_q(x)$, where $y\in[0,1]$. Again, by using a chain rule for $T(t,\cdot)_\#q=p_t$ with $p_t:=p(t,x)$, we have
\begin{equation}\label{cofv}
q(x)=\frac{dy}{dx}=\frac{1}{\frac{dx}{dy}}=\frac{1}{Q_q'(y)},\quad 
\partial_xT(t,x)=\partial_xF_{p_t}^{-1}(F_q(x))=Q_{p_t}'(y)\frac{dy}{dx}=\frac{Q'_{p_t}(y)}{Q'_q(y)}.
\end{equation}
Under the above change of variables and the fact that $\partial_t \log \partial_xT(t,x)=\partial_t\log Q'_{p_t}(y)$, we observe that the variation problem \eqref{variation1} satisfies  
\begin{equation}\label{variation}
\begin{split}
\mathrm{Dist}_{\mathrm{H}}(p^0,p^1)^2%=&\inf_{Q_{p_t}'\colon [0,1]^2\rightarrow\mathbb{R}}\Big\{\int_0^1\int_0^1|\partial_tQ_{p_t}'(y)\frac{1}{(Q_{p_t}'(y))^{\frac{3}{2}}}\sqrt{f''(\frac{1}{Q_{p_t}'(y)})}|^2dydt\Big\}\\
=&\inf_{Q'_{p_t}\colon [0,1]^2\rightarrow\mathbb{R}}\Big\{\int_0^1\int_0^1 |\partial_t \log Q'_{p_t}(y)|^2dydt\Big\},
\end{split}
\end{equation}
where the infimum is taken among all smooth paths $Q'_{p_t}\colon [0,1]^2\rightarrow\mathbb{R}$ with fixed initial and terminal time conditions. By using the Euler-Lagrange equation for variable $Q'_{p_t}$, we show that the geodesic equation in transport Hessian metric satisfies 
\begin{equation*}
\partial_{tt}\log Q'_{p_t}(y)=0. 
\end{equation*}
This means 
\begin{equation*}
\log(Q_{p_t}'(y))=t\log(Q_{p}'(y))+(1-t)\log(Q_{q}'(y)),
\end{equation*}
and 
\begin{equation*}
\partial_t\log(Q_{p_t}'(y))=\log(Q_{p}'(y))-\log(Q_{q}'(y)). 
\end{equation*}
Combining the above facts in the variational problem \eqref{variation1}, we finish the proof.
\end{proof}
\subsection{Transport KL divergence}
We also briefly present the derivation of transport KL divergence. Consider a functional $\mathcal{F}\colon\mathcal{P}\rightarrow\mathbb{R}$. We define a class of Bregman divergences in Wasserstein-$2$ space by 
\begin{equation*}
\mathrm{D}_{\rt,\mathcal{F}}(p\|q)=\mathcal{F}(p)-\mathcal{F}(q)-\int_\Omega \Big(\partial_x\frac{\delta}{\delta q(x)}\mathcal{F}(q), T(x)-x\Big)q(x)dx,
\end{equation*}
where $p$, $q\in \mathcal{P}(\Omega)$, $\frac{\delta}{\delta q(x)}$ is the $L^2$ first variation w.r.t. $q(x)$, and $T_\#q=p$. { See examples of $L^2$ first variation of functionals in \cite[10.4.2]{AGS}.} We refer $\mathrm{D}_{\rt,\mathcal{F}}$ as the transport Bregman divergence. If $\mathcal{F}$ is a second moment functional, i.e. $\mathcal{F}(p)=\int_\Omega |x|^2 p(x)dx$, then $\mathrm{D}_{\rt, \mathcal{F}}$ forms the Wasserstein-2 distance. If $\mathcal{F}(p)=-\mathcal{H}(p)=\int_\Omega p(x)\log p(x)dx$, then 
\begin{equation*}
\mathcal{F}(p)=-\mathcal{H}(p)=\int_\Omega \log \frac{q(x)}{T'(x)}q(x) dx. 
\end{equation*}
Thus, the transport Bregman divergence of $\mathcal{F}$ satisfies 
\begin{equation*}
\begin{split}
\mathrm{D}_{\rtKL}(p\|q)=&\mathrm{D}_{\rt,\mathcal{F}}(p\|q)\\
=& \int_\Omega \log \frac{q(x)}{T'(x)}q(x)-q(x)\log q(x)-(\partial_x\log q(x), T(x)-x)q(x)dx\\
=& \int_\Omega\Big[-\log T'(x)q(x)-(\partial_x q(x), T(x)-x)\Big]dx\\
=&\int_\Omega \Big(T'(x)-\log T'(x)-1\Big)q(x)dx,
\end{split}
\end{equation*}
where we use the fact that $q(x)\partial_x\log q(x)=\partial_xq(x)$ in the third equality, and apply the integration by parts in the last equality. From now on, we name $\mathrm{D}_{\rtKL}$ the {\em transport KL divergence}. Again, by the change of variable \eqref{cofv}, we note that $\mathrm{D}_{\rtKL}$ can be formulated in terms of quantile density functions:  
\begin{equation*}
\mathrm{D}_{\rtKL}(p\|q):=\int_0^1\Big(\frac{Q_{p}'(u)}{Q_{q}'(u)}-\log\frac{Q_{p}'(u)}{Q_{q}'(u)}-1\Big)du=\mathrm{D}_{\rt,1}(p\|q).
\end{equation*}
From this formulation, we observe that $\mathrm{D}_{\rtKL}$ is an Itakura--Saito type divergence in term of transport map functions or quantile density functions.

%%\subsection{Proof of Eulerian formulations for transport alpha geodesics}

\subsection{Validation of canonical divergences}
We are now ready to valid that the proposed transport alpha divergence is a canonical divergence, which follows the Definition 4.8 in \cite{IG2}. In other words, the following proposition holds. 
\begin{proposition}
Let $r(t,\cdot)$, $t\in [0,1]$, be the probability density, satisfying transport alpha geodesic in Proposition \ref{prop6}, with $\alpha=-1$, $r(0,x)=p(x)$ and $r(1,x)=q(x)$. 
Then the following equality holds: 
\begin{equation*}
\mathrm{D}_{\rt,1}(p\|q)=\int_0^1 t\cdot g_{\mathrm{H}}(\partial_t r, \partial_tr) dt,
\end{equation*}
where $g_{\mathrm{H}}$ is the transport Hessian metric defined in \eqref{THM}. 
\end{proposition}
\begin{proof}
We follow the idea of proof in \cite{IG2}, with using the change of variable \eqref{cofv} and the transport Hessian metric $g_{\mathrm{H}}$. 
%From Proposition \ref{prop6}, 
%when $\alpha=-1$, 
We denote a monotone function $\hat T(x)=T^{-1}(x)$, such that $(T^{-1})_\#p=q$. Denote $\hat T\colon [0,1]\times \Omega\rightarrow\Omega$ satisfying Proposition \ref{prop6} with $\alpha=-1$, such that 
$$\partial_t(\partial_x\hat T(t,x))=\partial_t(t\hat T'(x)+(1-t))=\hat T'(x)-1.$$ 
We prove the following claim. 

\noindent\textbf{Claim:} 
\begin{equation*}
\begin{split}
-\frac{d^2}{dt^2}\mathcal{H}(r(t,\cdot))=&g_{\mathrm{H}}(\partial_t r, \partial_tr).
%%=&\int_\Omega|\partial_t \partial_xT_1(t,x)\frac{1}{\partial_xT_1(t,x)}|^2{q(x)}dx.
\end{split}
\end{equation*}
\begin{proof}[Proof of Claim]
Similar as the derivation of equation \eqref{variation1}, we note that 
\begin{equation*}
\begin{split}
g_{\mathrm{H}}(\partial_t r, \partial_tr)=&\int_\Omega\left|\partial_t \partial_x\hat T(t,x)\frac{1}{\partial_x\hat T(t,x)}\right|^2{q(x)}dx.
\end{split}
\end{equation*}
We also note that $-\mathcal{H}(r(t,\cdot))=\int_\Omega r(t,x)\log r(t,x)dx=\int_\Omega p(x)\log\frac{p(x)}{\partial_x\hat T(t,x)} dx$. Hence
\begin{equation*}
\begin{split}
-\frac{d^2}{dt^2}\mathcal{H}(r(t,\cdot))=&\int_\Omega \frac{d^2}{dt^2}\log \frac{p(x)}{\partial_x\hat T(t,x)}p(x) dx\\ 
=&\int_\Omega\left(\left|\partial_t \partial_x\hat T(t,x)\frac{1}{\partial_x\hat T(t,x)}\right|^2-\frac{\partial_{tt}\partial_x\hat T(t,x)}{\partial_x\hat T(t,x)}\right){q(x)}dx\\
=&g_{\mathrm{H}}(\partial_t r, \partial_tr).
\end{split}
\end{equation*}
Here we use the fact that $\partial^2_{tt}\partial_x\hat T(t,x)=0$. 
\end{proof}
Thus, 
\begin{equation*}
\begin{split}
&\int_0^1 t\cdot g_{\mathrm{H}}(\partial_t r, \partial_tr) dt=\int_0^1 t\cdot (-\frac{d^2}{dt^2}\mathcal{H}(r(t,\cdot)))dt\\
=&t\cdot (-\frac{d}{dt}\mathcal{H}(r)|_{t=0}^{t=1})-\int_0^1(-\frac{d}{dt}\mathcal{H}(r))dt\\
=&-\frac{d}{dt}\mathcal{H}(r)|_{t=1}+\mathcal{H}(q)-\mathcal{H}(p). 
\end{split}
\end{equation*}
Here we note that 
\begin{equation}\label{a}
\begin{split}
-\frac{d}{dt}\mathcal{H}(r)|_{t=1}=&\int_\Omega p(x)\frac{d}{dt}\log\frac{p(x)}{\hat T'(x)} dx|_{t=1}\\
=&-\int_\Omega p(x) \partial_t\log \hat T'(t,x) dx|_{t=1}\\
=&-\int_\Omega p(x) \frac{\partial_t\hat T'(t,x)}{\hat T'(t,x)} dx|_{t=1}\\
=&-\int_\Omega p(x) \frac{\hat T'(x)-1}{\hat T'(x)} dx=-1+\int_\Omega \frac{1}{\hat T'(x)}p(x) dx. 
\end{split}
\end{equation}
Similar as in the change of variable formula \eqref{cofv}, we have 
\begin{equation*}
\int_\Omega \frac{1}{\hat T'(x)}p(x) dx=\int_0^1\frac{Q_p'(u)}{Q_q'(u)}du. 
\end{equation*}
We also note the fact that 
\begin{equation}\label{b}
\begin{split}
\mathcal{H}(q)-\mathcal{H}(p)=&\int_\Omega \Big[p(x)\log p(x)-q(x)\log q(x)\Big]dx \\
=&\int_\Omega p(x)\log p(x) dx-\int_\Omega \log q(\hat T(x)) q(\hat T(x))\hat T'(x)dx \\
=&\int_\Omega p(x)\log p(x) dx-\int_\Omega \log \frac{p(x)}{\hat T'(x)} p(x)dx \\
=&\int_\Omega p(x)\log \hat T'(x)dx\\
%=&\int_\Omega p(x)\log \frac{1}{\frac{d}{dF_p(x)}F_q^{-1}(F_p(x))}dx\\
=&-\int_0^1\log\frac{Q_p'(u)}{Q_q'(u)}du,
\end{split}
\end{equation}
where the last equality again follows from the change of variable formula in \eqref{cofv}. Combining the above equalities in \eqref{a} and \eqref{b}, we finish the proof. 
\end{proof}

}

\end{document}